\documentclass[11pt, draftcls,onecolumn]{IEEEtran}
\usepackage{multirow, verbatim, amsfonts, graphicx, amsmath, amsbsy, amssymb,  epsfig, url}

\usepackage{booktabs}

\newcommand{\ab}{{\mathbf a}}

\newcommand{\pb}{{\mathbf p}}
\newcommand{\qb}{{\mathbf q}}

\newcommand{\xb}{{\mathbf x}}

\renewcommand{\Re}{{\mathbb{R}}}

\newtheorem{definition}{Definition}

\newtheorem{remark}{Remark}
\newtheorem{theorem}{Theorem}

\newtheorem{corollary}{Corollary}

\title{Improving Noise Robustness in Subspace-based Joint Sparse Recovery}

\author{Jong Min Kim, Ok Kyun Lee and Jong Chul Ye}

\begin{document}
%
\maketitle
\begin{abstract}
\baselineskip 0.18in
In a multiple measurement vector  problem (MMV),  where multiple signals share a common sparse support and are sampled by a common sensing matrix,
we can expect joint sparsity to enable a further reduction in the number of required
measurements.  While a diversity gain from  joint sparsity had been demonstrated earlier in the case of  a convex relaxation method using an $l_1/l_2$ mixed norm penalty, 
only recently was it shown that similar diversity gain can be achieved by  greedy algorithms if we combine greedy steps with a MUSIC-like  subspace criterion. However,  the main limitation of these hybrid algorithms is that they often require  a
  large number of snapshots or a high signal-to-noise ratio (SNR)  for an accurate subspace as well as partial support estimation. 
One of the main contributions of this work is to show that  the noise robustness of  these algorithms  can be significantly improved by allowing sequential subspace estimation and  support filtering, even when the number of snapshots is insufficient.
Numerical simulations show that a novel sequential compressive MUSIC (sequential CS-MUSIC) that combines the sequential subspace estimation and support filtering steps significantly outperforms the existing greedy algorithms and  is quite comparable with computationally expensive state-of-art algorithms.
\end{abstract}

\IEEEpeerreviewmaketitle

\noindent{ 
\\

\noindent Correspondence to:\\
Jong Chul Ye,  Ph.D. ,
Associate Professor \\\
Dept. of Bio and Brain Engineering,  KAIST \\
373-1 Guseong-dong  Yuseong-gu, Daejon 305-701, Republic of Korea \\
Email: jong.ye@kaist.ac.kr \\
Tel: 82-42-350-4320 \\
Fax: 82-42-350-4310 \\
 }

\newpage
\baselineskip 0.29in

\section{Introduction}
\label{sec:intro}

We study a multiple measurement vector (MMV) problem,
where multiple signals share a same common sparse support set
and each signal is measured by multiplying it by a measurement
matrix.  
An MMV problem is one way in which multiple correlated signals can appear in a signal ensemble, and MMV problems also have many important applications \cite{KrVi96,PrWeScBo99,Joshi2006spi,Lee2011CDOT}.
A central theme in these studies has been that
joint sparsity within signal ensembles enables a further reduction in the number of required
measurements~\cite{BaWaDuSaBa05,Duarte2006IPSN},
where the number of measurements required per sensor must account for the
minimal features unique to that sensor~\cite{Kim2010CMUSIC, chen2006trs,cotter2005ssl,Mishali08rembo,Berg09jrmm,Lee2010SAMUSIC}.
Indeed,  for the case of an $l_1/l_2$ mixed norm approach, Obozinski {\em et al.} \cite{obozinski2011support} showed that a near optimal diversity gain can be achieved.

Recently,  Kim {\em et al.} \cite{Kim2010CMUSIC} and Lee {\em et al.} \cite{Lee2010SAMUSIC} independently showed that such a diversity gain can be also achieved in a new class of greedy algorithms by exploiting  the so-called   generalized (or extended) multiple signal classification (MUSIC) criterion \cite{Kim2010CMUSIC,Lee2010SAMUSIC}.  More specifically, these algorithms obtain a  partial support estimate using a conventional MMV greedy algorithm,  and then   the atoms corresponding to the partial supports are augmented into a data matrix to obtain an {augumented signal subspace} estimate. Finally,  a MUSIC-like \cite{Sc86}   criterion is derived for the  augmented subspace to find the remaining support. 
 The hybridization makes these hybrid greedy algorithms fully utilize a diversity gain so that the algorithms outperform all the existing greedy methods.

The performance improvement of these greedy algorithms is substantial and nearly achieves the $l_0$ bound when a signal subspace and partial support estimation  are  accurate due to a sufficient number of snapshots or  high signal to noise ratio (SNR) \cite{Kim2010CMUSIC,Lee2010SAMUSIC}. However,  if either of these estimation is erroneous owing to an insufficient number of snapshots or low SNR,   performance degrades. 
Similar observations have been made  in the literature on classical array signal processing \cite{linebarger1995incorporating,wirfalt2011prior}.  In \cite{linebarger1995incorporating},  prior knowledge of the direction-of-arrival (DOA) has been incorporated to improve the performance of MUSIC by filtering out the known sources via orthogonal projections. However,  as shown in \cite{wirfalt2011prior}, such orthogonal projection is suboptimal from a statistical standpoint.

While increasing the number of snapshots is relatively easier in classical sensor array signal processing problems,   in  some MMV problems such as parallel MR imaging \cite{PrWeScBo99},  an additional snapshot requires a hardware change by adding a new  receiver coil. Hence, in these problems, exploiting  other dimensions  would be beneficiall.
We are aware that joint sparse recovery methods such as  Bayesian approaches
\cite{wipf2007ebs,stoica2012LIKES}, or convex optimization techniques \cite{stoica2011spice}, are shown to be statistically robust in the direction of arrival estimation problems, as first demonstrated by Malitov {\em et al.} \cite{malioutov2005ssr} and further developed by Stoica {\em et al.} \cite{stoica2012LIKES}.   
However, these approaches are usually computationally expensive for  MMV problems with a large number of sensors, so we need a new greedy algorithm that achieves a similar optimal performance with a significantly reduced computational complexity.

Therefore, one of the main goals of this paper is to address how these hybrid greedy methods can be  made robust without
increasing the number of snapshots.
One important  contribution  is a  new theory explaining that the generalized MUSIC criterion is a special case of a new subspace criterion that can 
be used to derive  two  sequential strategies  to  improve the accuracy of an augmented signal subspace estimation. More specifically,  a forward greedy subspace estimation step improves the robustness of an augmented signal subspace estimation by adding newly discovered atoms in the MUSIC step, whereas the backward support filtering provides additional robustness by eliminating the inaccurate portion of support estimates.   
 By combining the two steps,  we develop a novel sequential CS-MUSIC algorithm that is robust, even with a limited number of snapshots. 
 Using theoretical noise analysis as well as 
numerical simulation, we show that the sequential CS-MUSIC is superior to the existing subspace-based greedy algorithms and exhibits  similar performance behavior to the mixed-norm  \cite{obozinski2011support,malioutov2005ssr} or Bayesian approaches \cite{wipf2007ebs} with a significantly reduced computationally complexity.

\section{Generalized Subspace Criterion}

\subsection{Notations and Mathematical Preliminaries}

Throughout the paper, $\xb^i$ and $\xb_j$ correspond to the $i$-th row and $j$-th column of matrix $X$. When $I$ is an index set, $X^I$, $A_I$ corresponds to a submatrix collecting corresponding rows of $X$ and columns of $A$, respectively. 
%
The rows (or columns) in $\mathbb{R}^n$ are in {\em general position} if any $n$ collection of rows (or columns) are linearly independent.

\begin{definition}[Canonical form noiseless MMV \cite{Kim2010CMUSIC}]\label{def:can}
 Suppose 
we are given a sensing matrix $A\in\mathbb{R}^{m\times n}$ and an
observation matrix $B\in \mathbb{R}^{m\times r}$ such that
$B=AX_{*}$ for some $X_{*}\in\mathbb{R}^{n\times r}$ and
$\|X_{*}\|=|{\rm supp}X|=k$,
where $m$, $n$, and $r$ are positive integers ($r\leq m<n$) that
represent the number of sensor elements, an ambient space
dimension, and the rank of an observation matrix, respectively.
 A canonical form noiseless multiple
measurement vector (MMV) problem is given as an estimation problem of
$k$-sparse vectors $X\in\mathbb{R}^{n\times r}$ using the following formula:
\begin{eqnarray}\label{eqdefcan_mmv}
{\rm minimize}~~~\|X\|_0\\
{\rm subject~to}~~~B=AX, \notag
\end{eqnarray}
where $\|X\|_0=|{\rm supp}X|$, ${\rm supp}X=\{1\leq i\leq n :
\mathbf{x}^i\neq 0\}$, $\mathbf{x}^i$ is the $i$-th row of $X$, and
the observation matrix $B$ is full rank, i.e. ${\rm rank}(B)=r\leq
k$.
\end{definition}

Recall that every MMV problem can be converted to a canonical form MMV by using a singular value decomposition and dimension reduction as described in \cite{Kim2010CMUSIC}. Hence, in this section, we assume that an MMV problem assumes the canonical form. However, this assumption will be relaxed later in noise analysis.

%
\subsection{Generalized Subspace Criterion}

Note that the generalized MUSIC criterion in \cite{Kim2010CMUSIC} requires $0\leq \delta_{2k-r+1}^L(A)<1$.
This implies that, if a sensing matrix is  obtained from a random Gaussian and if measurement is noiseless,  then we have the following minimal sampling condition \cite{Kim2010CMUSIC}:
\begin{equation*}
m\geq (1+\delta)(2k-r+1)~{\rm for~some}~\delta>0.
\end{equation*}
If we have  a redundant sampling $m\gg 2k-r+1$,  the following theorem can be used instead as the extension of  the generalized MUSIC criterion in  \cite{Kim2010CMUSIC}.

\begin{theorem} 
\label{thm:genmusic1}
Suppose $1\leq l\leq r$ and we have a canonical MMV model $AX=B$  with a sensing matrix  $A$ that
satisfies  an RIP condition with $0\leq\delta_{2k-r+l}^L(A)<1$. Furthermore, suppose the nonzero rows of $X$ are in general position. Then, for a given  index set $I\subset \{1,\cdots,n\}$ such that $|I|\leq \min(2(k-r)+l,k)$ and $|I\setminus {\rm supp}X|\leq k-r+l$,    the following statements are equivalent:
\begin{eqnarray}\label{rank-def}
({\rm i})&&|I\cap {\rm supp}X|\geq k-r+1;\notag\\
({\rm ii})&&{\rm rank}[A_I~B]<|I|+r .  
\end{eqnarray}
\end{theorem}
\begin{proof}
(i)$\Longleftrightarrow$(ii): Assume that $|I\cap {\rm supp}X|\geq k-r+1$. Then $|I|+r\leq  2k-r+l\leq m$ since $\delta_{2k-r+l}^L(A)<1$. If we take $\tilde{I}\subset (I\cap {\rm supp}X)$ such that $|\tilde{I}|=k-r+1$, then
$${\rm dim}(R[A_{\tilde{I}}~B])\leq {\rm dim}(A_S)=k.$$
However, $|\tilde{I}|+r=k+1$ so that $[A_{\tilde{I}}~B]$ is not of full column rank. Hence $[A_I~B]$ is not also of full column rank since $\tilde{I}\subset I$.\\
Conversely, if we assume \eqref{rank-def}, there are $\pb\in\mathbb{R}^{|I|}$ and $\qb\in\mathbb{R}^r$ such that $A_I\pb+AX\qb={\bf 0}$ and $[\pb, \qb]^T\neq {\bf 0}$. If we let
$\tilde{\pb}\in\mathbb{R}^n$ by $\tilde{\pb}^I=\pb$ and $\tilde{\pb}^{I^c}={\bf 0}$,  we have
$A[\tilde{\pb}+X\qb]={\bf 0}.$ Since $\|\tilde{\pb}+X\qb\|_0\leq |I\setminus{\rm supp}X|+|{\rm supp}X|\leq 2k-r+l$, by the RIP condition, we have $\tilde{\pb}+X\qb={\bf 0}$ so that ${\rm supp}(\tilde{\pb})={\rm supp}(X\qb)\subset {\rm supp}X.$ If we assume that $|I\cap{\rm supp}X|\leq k-r$, then $\|\tilde{\pb}\|_0\leq k-r$ but $\|X\qb\|_0\geq k-r+1$ since the nonzero rows of $X$ are in general position. This is impossible so that $|I\cap {\rm supp}X|\geq k-r+1$.
\end{proof}

\bigskip

Note that  the conditions $|I|\leq \min(2(k-r)+l,k)$ and $|I\setminus {\rm supp}X|\leq k-r+l$ in Theorem~\ref{thm:genmusic1} do not  imply that there is a unique  index set $I$; rather, Theorem \ref{thm:genmusic1} says that multiple index sets $I$ can exist for a given $l$. For example, if $l=1$,
any index set $I$ such that $|I|=k-r+1,\cdots, \min(2(k-r)+1,k)$ that satisfies the condition $|I\setminus {\rm supp}X|\leq k-r+l$, can be used to test conditions (i)-(ii) in Theorem \ref{thm:genmusic1}. Furthermore, if we choose $|I|=k-r+1$,  Theorem~\ref{thm:genmusic1} is reduced to the following generalized MUSIC criterion  in \cite{Kim2010CMUSIC}.
%
\begin{corollary}[Generalized MUSIC Criterion \cite{Kim2010CMUSIC}]
Suppose we have a canonical MMV model $AX=B$  with a sensing matrix  $A$ that
satisfies  an RIP condition with $0\leq\delta_{2k-r+1}^L(A)<1$. Furthermore, suppose the nonzero rows of $X$ are in general position. 
Then, for $I_{k-r}\subset {\rm supp}
X$ with $|I_{k-r}|=k-r$ and any $j\in\{1,\cdots,n\}\setminus I_{k-r} $, we have $j\in {\rm supp}X$ if and only if
\begin{equation}\label{eq:gemusic}
{\rm rank}[A_{I_{k-r}\cup\{j\}}~B]<k+1
\end{equation}
or equivalently 
\begin{equation*}
\ab_j^{*}P_{R([A_{I_{k-r}}~B])}^{\perp}\ab_j=0.
\end{equation*}
\end{corollary}
\begin{proof}
For an index set $I$ such that $|I|=k-r+1$, the condition $|I\setminus {\rm supp}X| \leq k-r+1$ always holds. Therefore,  $|I\cap {\rm supp}X|= k-r+1$ is equivalent to $I \subset {\rm supp}X$. Hence, for  $I = I_{k-r} \cup \{ j\}$ such that  $I_{k-r}\subset {\rm supp}X$ and $|I_{k-r}|=k-r$,  Eq.~\eqref{rank-def} is equivalent to  Eq.~\eqref{eq:gemusic}, which is equivalent to say $\ab_j \in R([A_{I_{k-r}} B])$ or $\ab_j^*P^\perp_{R([A_{I_{k-r}} B])}\ab_j = 0$. This concludes the proof.
\end{proof}

\begin{remark}
If $r=k$,   the  conventional MUSIC criterion can be trivially derived.
\end{remark}
\begin{remark}
The subspace $R([A_{I_{k-r}}~B])$ is called {\em augmented signal subspace}. This name was first coined in \cite{Lee2010SAMUSIC}.
\end{remark}

\noindent So far, we have shown that Theorem~\ref{thm:genmusic1} can reproduce the existing results. However, one of the important byproducts of the theorem is the    following form, which will be used extensively in the following sections.
\begin{corollary}\label{prop-rank}
Suppose $1\leq l\leq r$ and we have a canonical MMV model $AX=B$  with a sensing matrix  $A$ that
satisfies  an RIP condition with $0\leq\delta_{2k-r+l}^L(A)<1$. Furthermore, suppose the nonzero rows of $X$ are in general position. 
Then, for an index set $I\subset \{1,\cdots,n\}$ such that $|I|\leq \min(2(k-r)+l,k)$ and $|I\setminus {\rm supp}X|\leq k-r+l$,  if we have
$|I\cap {\rm supp}X|=k-r+q$ for some $q\geq 0$, then
\begin{equation}\label{rank-def2}
{\rm rank}[A_I~B]=|I|+r-q.
\end{equation}
\end{corollary}
\begin{proof}
Take an $I_{k-r}\subset (I\cap {\rm supp}X)$ with $|I_{k-r}|=k-r$ and let $J_{k-r}=(I\cap {\rm supp}X)\setminus I_{k-r}$, where $|J_{k-r}|=q$. Then by Theorem \ref{thm:genmusic1}, we have
$${\rm rank}[A_{I\setminus J_{k-r}}~B]=|I\setminus J_{k-r}|+r=|I|+r-q.$$
Then for any $j\in J_{k-r}$, $\ab_j\in R([A_{I_{k-r}}~B])=R(A_S)$ so that
$${\rm rank}[A_I~B]={\rm rank}[A_{I\setminus J_{k-r}}~B]=|I|+r-q$$
since $R([A_I~B])=R([A_{I\setminus J_{k-r}}~B]).$
\end{proof}

\section{Sequential Compressive MUSIC Algorithm}

By employing the results in  the previous section, this section first develops  forward or backward greedy steps. Then,  by combining the two approaches,  we can derive a  novel {\em sequential CS-MUSIC} algorithm.

\subsection{Forward Greedy: Sequential Subspace Estimation}

In \cite{Kim2010CMUSIC}, the CS-MUSIC first determines $k-r$ indices of supp$X$ with CS-based algorithms such as 2-thresholding or S-OMP, and then it recovers the  remaining $r$ indices of supp$X$ using the  generalized MUSIC criterion. For this,  a projection operator onto the noise subspace is calculated  as the orthogonal complement of  the augmented signal subspace $R([A_{I_{k-r}}~B])$. However,  the following result can further extend the existing generalized MUSIC criterion  \cite{Kim2010CMUSIC} .  


\begin{theorem}\label{thm-gmusic-add}
Suppose $1\leq l\leq r$ and we have a canonical MMV model $AX=B$  with a sensing matrix  $A$ that
satisfies  an RIP condition with $0\leq\delta_{2k-r+l}^L(A)<1$. Furthermore, suppose  the nonzero rows of $X$ are in general position. 
%
Then, if we have an index set $I\subset \{1,\cdots,n\}$ such that $|I|\leq \min(2(k-r)+l-1,k-1)$, $|I\setminus{\rm supp}X|\leq k-r+l-1$ and $|I\cap {\rm supp}X|\geq k-r$, we have for $j\notin I$,
$j\in {\rm supp}X$ if and only if
\begin{equation}\label{forward-rank}
{\rm rank}[A_{I}~B]={\rm rank}[A_{I \cup \{j\}}~B]
\end{equation}
or equivalently 
\begin{equation}\label{forward-back2}
\ab_j^{*}P_{R([A_I~B])}^{\perp}\ab_j=0.
\end{equation}
\end{theorem}
\begin{proof}
By the condition we have $|I\cup\{j\}|\leq \min(2(k-r)+l,k)$ and $|(I\cup\{j\})\setminus {\rm supp}X|\leq k-r+l$ so that we can apply Corollary \ref{prop-rank} for $I$ and $I\cup \{j\}$ since $|I\cap {\rm supp}X|\geq k-r$.
If $|I\cap {\rm supp}X|=k-r+q$ for some $q\geq 0$, then we have ${\rm rank}[A_I~B]=|I|+r-q$.
Then for any $j\notin I$, if we have $j\in {\rm supp}X$, $|(I\cup\{j\})\cap {\rm supp}X|=k-r+(q+1)$ so that we have
$${\rm rank}[A_{I\cup \{j\}}~B]=|I|+1+r-(q+1)=|I|+r-q={\rm rank}[A_I~B].$$
On the other hand, if we have $j\notin {\rm supp}X$, $|(I\cup\{j\})\cap {\rm supp}X|=k-r+q$ so that we have
$${\rm rank}[A_{I\cup\{j\}}]=|I|+1+r-q>{\rm rank}[A_I~B].$$
Finally, \eqref{forward-rank} is equivalent to $\ab_j\in R([A_I~B])$, which is also equivalent to \eqref{forward-back2}. This completes the proof.
\end{proof}

\begin{remark}
Note that $R([A_I~B])=R([A_{I_{k-r}}~B])$ and ${\rm dim}R([A_{I_{k-r}}~B])=k$ for all $I\subset {\rm supp}X$ and $k-r\leq |I|\leq \min{(2(k-r)+l-1,k-1)}$. This implies that we first need to find $I_{k-r}$ support using a compressive sensing algorithm, then  we augment  newly added supports into the initial estimate $I_{k-r}$.
 As will be shown later in noise analysis, such  a greedy procedure improves the accuracy of the augmented signal subspace estimation. 
 \end{remark}
 \begin{remark}
The greedy procedure can even be performed in a critically sampled case, {\em i.e.} $l=1$.  In this case,  we can augment atoms up to $\min(2(k-r),k-1)$, which is always bigger than adding only $k-r$ atoms. However,  the number of possible augmentation increases with a redundant sampling, which makes the algorithm more robust.
 \end{remark}
 
 Theorem~\ref{thm-gmusic-add} leads us to the following sequential algorithm  ({\sf SeqSubspace}), as in Table I. Note that the algorithm can be combined with any joint sparse recovery algorithm that provides a $k-r$ initial support estimate. 
%

\begin{table}[!h]
 \center{\caption{}
 \begin{tabular}{l}
\toprule
Algorithm: $I$= {\sf SeqSubspace($A, B, I_{k-r}$)}\\
\midrule
~- Set $q=0$ and $I=I_{k-r}$. \\
~- While $q<r$, do the following procedure:\\
~~1. Perform an SVD of $[A_I~B]=[U_1,U_0]{\rm diag}[\Sigma_1,\Sigma_0][V_1,V_0]^{*}$, \\ ~~~~where $\Sigma_1={\rm diag}[\sigma_1,\cdots,\sigma_k]$
 and $\Sigma_0={\rm diag}[\sigma_{k+1},\cdots,\sigma_{k+q}]$  \\~~~~and $\sigma_1\geq\sigma_2\geq\cdots\geq\sigma_{k+q}$.\\
~~2. Take $j_q=\arg\min_{j\notin I}\|P_{R(U_1)}^{\perp}\ab_j\|^2.$\\
~~3. Set $I:=I\cup \{j_q\}$, let $q:=q+1$ and goto step 1. \\
~- Return $I$. \\
\bottomrule
\end{tabular}}
\end{table}

\subsection{Backward Greedy: Support Filtering}

As discussed before,  we can easily expect that the performance of the generalized MUSIC step is highly dependent on the selection of $k-r$ correct indices of the support of $X$.  Note that this is a very stringent condition.  In practice, even though the first consecutive steps of, for example, S-OMP, may not provide all true partial supports, it is more likely that among a $k$-sparse support estimate of S-OMP, part of the supports (not in sequential order) can be correct.
In fact,  an information theoretical analysis of a  partial support recovery condition in single measurement vector CS (SMV-CS) \cite{Reeves2008ISIT} showed that the required SNR condition of a partial support recovery is much more relaxed than that for  a full support recovery.
 Hence, if the estimate of the support of $X$ has at least $k-r$ indices of the support of $X$ and we can identify them, then we can expect that the performance of the compressive MUSIC will be improved. When ${k \choose k-r}$ is small, we may apply the exhaustive search, but if both $k-r$ and $r$ are not small, then the exhaustive search is hard to apply so that we have to find some alternative method to identify correct indices from an  estimate of ${\rm supp}X$.

Indeed, 
our new algorithm requires that
 $k-r+1$ supports (not in sequential order) out of  a larger support estimate is correct.  Then, the location of a  correct $k-r$ support  can be readily estimated using the following backward support filtering. Compared to a forward greedy procedure that improves the accuracy of the signal subspace estimation,  the backward support filtering criterion can  improve the accuracy of a partial support recovery, and, hence, the corresponding accuracy of an augmented signal subspace.
  
\begin{theorem}[Backward support filtering criterion]\label{thm-gmusic-sub}
Suppose $1\leq l\leq r$ and we have a canonical MMV model $AX=B$  with a sensing matrix  $A$ that
satisfies  an RIP condition with $0\leq\delta_{2k-r+l}^L(A)<1$. Furthermore, suppose the nonzero rows of $X$ are in general position. 
%
Then, if we have an index set $I\subset \{1,\cdots,n\}$ such that $|I|\leq \min(2(k-r)+l,k)$, $|I\setminus {\rm supp}X|\leq k-r+l$ and $|I\cap {\rm supp}X|\geq k-r+1$, then we have for $j\in I$,
$j\in {\rm supp}X$ if and only if
$${\rm rank}[A_{I\setminus\{j\}}~B] = {\rm rank}[A_I~B],$$
or equivalently 
$$\ab_j^{*}P_{R([A_{I\setminus\{j\}}~B])}^{\perp}\ab_j=0.$$
\end{theorem}
\begin{proof}
Assume that $|I\cap {\rm supp}X|=k-r+q$, where $q\geq 1$. Then, by Corollary \ref{prop-rank}, we have ${\rm rank}[A_I~B]=|I|+r-q$. Noting that $I\setminus \{j\}$ satisfies the assumptions of Corollary \ref{prop-rank} for any $j\in I$, if we have $j\in {\rm supp}X$, $|I\cap {\rm supp}X|=k-r+(q-1)$ so that we have
$${\rm rank}[A_{I\setminus\{j\}}~B]=|I|-1+r-(q-1)=|I|+r-q={\rm rank}[A_I~B].$$
On the other hand, if we have $j\notin {\rm supp}X$, $|I\cap {\rm supp}X|=k-r+q$ so that we have
$${\rm rank}[A_{I\setminus\{j\}}~B]=|I|-1+r-q=|I|+r-q-1<{\rm rank}[A_I~B].$$
Finally, due to the rank condition, we know $\ab_j\in R([A_{I\setminus\{j\}}~B])$ if and only if $j\in {\rm supp}X$. Hence $\ab_j^{*}P_{R([A_{I\setminus\{j\}}~B])}^{\perp}\ab_j=0$ if and only if $j\in {\rm supp}X$.
That completes the proof.
\end{proof}

Theorem~\ref{thm-gmusic-sub} informs  us that if we have a partial estimate of support of $X$ that has at least $k-r+1$ correct indices of support of $X$, we can identify the correct part of the estimated partial support of $X$ by using the backward support filtering criterion  as described in Table II.

\begin{table}[!h]
 \center{\caption{}
 \begin{tabular}{l}
\toprule
Algorithm: $I_{k-r}$= {\sf SupportFiltering($A, B, I$)}\\
\midrule
~- For all $j\in I$, calculate the quantities $\zeta(j):=\|P_{R([A_{I\setminus\{j\}}~B])}^{\perp}\ab_j\|^2$.
\\
~- Making an ascending ordering of $\zeta(j)$ for $j\in I$, choose indices that \\~~~ correspond to the first $k-r$ indices 
 and put these indices into $I_{k-r}$. \\
~-  Return $I_{k-r}$. \\
\bottomrule
\end{tabular}}
\end{table}

%
%
%

\begin{remark}
Due to the condition $|I|\leq \min(2(k-r)+l,k)$ in Theorem~\ref{thm-gmusic-sub},   we can include $k$-sparse support estimate $I$ in a support filtering step if $r<(k+l)/2$. Note that this is always true regardless of $r$ if $l=k$ or  $\delta^L_{2k}<1$.  
 However,   if $r\geq (k+l)/2$,  we can use the following heuristics. 
First, just include the first $2(k-r)+l$ support estimate of $I$ for a support filtering. Since, in most  greedy algorithms, the earlier greedy steps are more likely to succeed,  correct $k-r+1$ supports are more likely to be included.  Hence, we can filter out the remaining indices $j\in I$ such that $j \notin {\rm supp}X$.
\end{remark}

\subsection{Sequential CS-MUSIC}

By combining the forward and the backward greedy steps, this paper develops the following sequential CS-MUSIC algorithm decribed in Table III. Note that this algorithm  assumes that the sparsity level $k$  is given as  {\em a priori} knowledge.  (The estimation problem of an unknown $k$  will be discussed later.)

\begin{table}[!h]
 \center{\caption{}
 \begin{tabular}{l}
\toprule
Algorithm: $I_k$= {\sf SeqCSMUSIC($A, Y, k,r$)}\\
~~~~~Input: ~$k,r$, $A \in \Re^{m\times n}$ ,  $Y\in \Re^{m\times N}$\\
~~~~~Output: $k$ support estimate $I_{k}$\\
\midrule
~- Estimate the $k$ support estimate $I_k$ of supp$X$ using {\em any} MMV algorithm.\\
~- $U~~~:= $ Rank-$r$ signal subspace estimate of $R(Y)$.\\
~- $I_{k-r}:=${\sf SupportFiltering}($A,U, I_k$).\\
~-$~I_k~~~:=${\sf SeqSubspace}($A,U,I_{k-r}$).\\
~-  Return $I_k$. \\
\bottomrule
\end{tabular}}
\end{table}

\section{Noisy Performance Analysis of Sequential CS-MUSIC}
\label{sec:proposedMethod}

\subsection{Improving Noise Robustness Using Sequential Subspace Estimation}

In practice,  measurements are noisy, so the theory we have derived for noiseless measurements should be modified.
Suppose a noisy MMV model  is given by:
$$Y= AX+W,$$ 
where $Y\in \Re^{m\times N}$ are noisy measurements corrupted by an additive noise $W\in \Re^{m\times N}$, and $N$ denotes the number of snapshots.
Then, using singular value decomposition, we can find the  following  canonical  MMV problem:
$$\tilde S = A \tilde X + \tilde W, $$
where 
 $\tilde{S}\in\mathbb{R}^{m\times r}$ is the rank-$r$ signal subspace estimate  of $Y$ and $r\leq N $ denotes the numerical rank of $Y$. Due to the noise,
 $\tilde{S}$ is peturbed from the noiseless signal subspace $S$ such that $R(S)=R\left(A\tilde X\right)$,
%
%
%
%
%
 which leads to errors in the augmented signal subspace.
 The following theorem characterizes how much perturbation in an augmented signal subspace can be endured by a  generalized MUSIC step.
 
 \begin{theorem}\label{thm:gmusic}
\ For $0\leq l<r$, if we have $I_{k-r+l}\subset {\rm supp}X$ such that $|I_{k-r+l}|=k-r+l$ and singular value decomposition of   $[A_{I_{k-r}} ~~S]$ and $[A_{I_{k-r+l}}~\tilde S]$ as
$$[A_{I_{k-r}} ~~S]
= U_1\Sigma_1 V_1, \quad [A_{I_{k-r+l}}~\tilde S]=\tilde{U}_1\tilde{\Sigma}_1\tilde{V}_1^{*}+\tilde{U}_0\tilde{\Sigma}_0\tilde{V}_0^{*},$$
where $\tilde{\Sigma}_1={\rm diag}[\tilde{\sigma}_1,\cdots,\tilde{\sigma}_k]$, $\tilde{\Sigma}_0={\rm diag}[\tilde{\sigma}_{k+1},\cdots,\tilde{\sigma}_{k+l}]$ and $\tilde{\sigma}_1\geq \tilde{\sigma}_2\geq\cdots\geq\tilde{\sigma}_{k+l}$, 
then,  for any $j\notin {\rm supp}X$ and $q\in {\rm supp}X$ we have
\begin{eqnarray*}
\|P_{R(\tilde{U}_1)}^{\perp}\ab_j\|^2&>& \|P_{R(\tilde{U}_1)}^{\perp}\ab_q\|^2
\end{eqnarray*}
provided that
\begin{equation}\label{pert-aug}
\|P_{R(\tilde{U}_1)}-P_{R(U_1)}\|\ < 1-\gamma \  .
\end{equation}
and $m,n \rightarrow \infty$ and
 $\gamma=\lim_{n\rightarrow\infty}k/m<1$ . In other words, a generalized MUSIC step finds   correct supports if Eq.~\eqref{pert-aug} is satisfied.
\end{theorem}
\begin{proof}
Noting that $\|P_{R(U_1)}^{\perp}\ab_j\|^2=0$ for $j\in {\rm supp}X$ by the generalized MUSIC criterion, for any $j\notin {\rm supp}X$ and $q\in {\rm supp}X$, we have
\begin{eqnarray}\label{eq:diff}
&&\|P_{R(\tilde{U}_1)}^{\perp}\ab_j\|^2-\|P_{R(\tilde{U}_1)}^{\perp}\ab_q\|^2 
\notag\\
&=&\ab_j^{*}P_{R(\tilde{U}_1)}^{\perp}\ab_j-\ab_q^{*}P_{R(U_1)}^{\perp}\ab_q
-\ab_q^{*}\left[P_{R(\tilde{U}_1)}-P_{R(U_1)}\right]\ab_q \notag\\
&=&\ab_j^{*}P_{R(\tilde{U}_1)}^{\perp}\ab_j-\|P_{R(\tilde{U}_1)}-P_{R(U_1)}\|\|\ab_q\|^2.
\end{eqnarray}
Since $a_{i,j}$'s are i.i.d. normal distribution with zero mean and variance $1/m$ and $\ab_j$ is independent of $P_{R(\tilde{U}_1)}^{\perp}$ for any $j\notin {\rm supp}X$, $\ab_j^{*}P_{R(\tilde{U}_1)}^{\perp}\ab_j$ is a chi-squared random variable of degree of freedom $m-k$ since ${\rm rank}(\tilde{U}_1)=k$. Also, for each $1\leq j\leq n$, $m\|\ab_j\|^2$ is a chi-squared random variable with degree of freedom $m$ so that we have, by Lemma 3 in \cite{Fletcher2009nscond},
$\lim\limits_{n\rightarrow\infty}{\max\limits_{1\leq j\leq n}\|\ab_j\|^2}/{m}=1$
since $\lim_{n\rightarrow\infty}(\log{n})/m=0$.
Since $\lim_{n\rightarrow\infty}(\log{(n-k)})/(m-k)=0$, by Lemma 3 in \cite{Fletcher2009nscond}, we have
$\lim\limits_{n\rightarrow\infty}{\min\limits_{j\notin{\rm supp}X}m\ab_j^{*}P_{R(\tilde{U}_1)}^{\perp}\ab_j}/{(m-k)}=1$ so that
$\lim\limits_{n\rightarrow\infty}{\min\limits_{j\notin {\rm supp}X}\ab_j^{*}
P_{R(\tilde{U}_1)}^{\perp}\ab_j}/{\max\limits_{1\leq j\leq n}
\|\ab_j\|^2}=
1-\gamma.$
Hence, Eq.~\eqref{eq:diff} is positive provided that Eq.~\eqref{pert-aug} holds
in the large system limit. This completes the proof.
\end{proof}

Therefore, by minimizing the perturbation in the augmented signal subspace $\|P_{R(\tilde{U}_1)}-P_{R(U_1)}\|$, we can make the generalized MUSIC step more robust. 
Unfortunately, the  direct minimization of the perturbation of the subspace is not easy.  Instead,  we are interested in minimizing the following upper-bound of the perturbation, whose proof can be found in Appendix~A:
\begin{equation}\label{eq:upper}
\|P_{R(\tilde{U}_1)}-P_{R(U_1)}\|\leq \frac{\Delta}{\sigma_k([A_{I_{k-r+l}}~S])-\Delta}  \ ,
\end{equation}
where $\Delta = \|S-\tilde S \|$ and $\sigma_k([A_{I_{k-r+l}}~S])$ denotes the $k$-th largest singular value of $[A_{I_{k-r+l}}~S]$.
Then, we have the following theorem:
\begin{theorem}\label{thm:snr}
Let $l\geq 0$. For  $I_{k-r+l}\subset {\rm supp}X$ such that $|I_{k-r+l}|=k-r+l$,   the generalized MUSIC steps find the remaining $r-l$ support
provided that
\begin{equation}\label{snr-aug}
\frac{\sigma_k([A_{I_{k-r+l}}~S])}{\Delta}>1+\frac{1}{1-\gamma}
\end{equation}
and  $m,n \rightarrow \infty$ and
 $\gamma=\lim_{n\rightarrow\infty}k/m<1$.
\end{theorem}
\begin{proof}
This can be trivially proven by pugging  Eq.~\eqref{eq:upper} in  the inequality \eqref{pert-aug}.
\end{proof}

 Note that for $I_{k-r+l}\subset {\rm supp}X$, we have ${\rm rank}([A_{I_{k-r+l}}~S])=k$ so that the set of columns of $[A_{I_{k-r+l}}~S]$ is a frame in $R(A_{{\rm supp}X})$ with lower frame bound $\sigma_k^2([A_{k-r+l}~S])$. In this case,  as shown in Fig.~\ref{fig:snr_seq},  $\sigma_k([A_{I_{k-r+l}}~B])$ is an increasing function of $l$, so as $l$ increases, the frame becomes more redundant and the lower frame bound become larger. Hence,  the left side of Eq.~\eqref{snr-aug} becomes larger.
 \begin{figure}[htbp]
\centering\includegraphics[width=8cm]{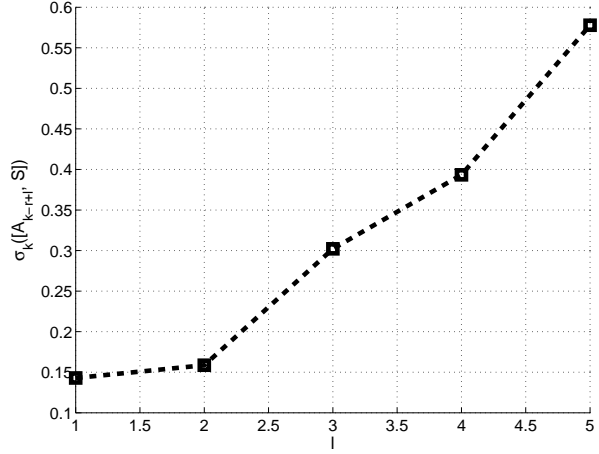}
\caption{$\sigma_k\left([A_{k-r+l}, S]\right)$ values  with increasing $l$. Simulation parameters are $n=128, k=8, r=6$. 
}
\label{fig:snr_seq}
\end{figure}

    This observation provides us an important error correction scheme. 
    Note that the SNR condition Eq.~\eqref{snr-aug}  is still the same even if we find a support  index $j_l$  in a greedy manner as follows:
 \begin{equation}\label{supp-id}
j_l=\arg\min\limits_{j\notin I_{k-r+l}}\|P_{R(\tilde{U}_1)}^{\perp}\ab_j\|^2 \ .
\end{equation} 
However, as SNR condition Eq.~\eqref{snr-aug}  is a sufficient condition, a non-zero probability of $j_l $ being in 
the true support    exists  even though Eq.~\eqref{snr-aug} is not satisfied. (This is especially true if we select only one index rather than choosing all $r-l$ indices).   If  a {\em correctly} found index $j_l$ is augmented for the next step of sequential subspace estimation, then it is more likely that the condition in Eq.~\eqref{snr-aug} can be satisfied in the following greedy steps since the left side term of  Eq.~\eqref{snr-aug} is an increasing function of $l$ thanks to the inclusion of a correct index $j_l$.  As soon as an SNR condition is satisfied,  the remaining greedy steps will succeed since the condition is sufficient.  Therefore, even when a sufficient SNR condition is not satisfied initially, the proposed sequential subspace estimation  technique exploits the possibility of finding a correct index to improve the noise robustness,  which was not possible in an  original MUSIC step.

\subsection{Improving Noise Robustness Using  Support Filtering}

Using similar techniques, we can derive the following sufficient condition for the success of  support filtering.

\begin{theorem}\label{thm:snr-back}
Let $m\geq k+r$. Suppose we have an index set $I\subset \{1,\cdots,n\}$ such that $|I| = k$, 
%
$|I_C|\geq k-r+1$, where $I_C:=({\rm supp}X)\cap I$. Then, we have
\begin{equation}\label{back-success}
\min\limits_{j\notin{\rm supp}X}\|P_{R[A_{I\setminus\{j\}}~\tilde{S}]}^{\perp}\ab_j\|^2>
\max\limits_{j\in \{q_1,\cdots,q_{k-r}\}}\| P_{R[A_{I\setminus\{j\}}~\tilde{S}]}^{\perp}\ab_j\|^2
\end{equation}
provided that
\begin{equation}\label{snr-back}
\frac{\tilde{\sigma}_k(q_{k-r},I)}{\Delta}> 1+
\frac{1}{1-\gamma(1+\alpha)},
\end{equation}
and $m,n \rightarrow \infty$, 
where $\tilde{\sigma}(q,I)=\max\{\sigma_k([A_{I_T}~S]):T\subset I_C\setminus\{q\},|T|=k-r\}$ for $q\in I_C$ and $I_C=\{q_1,\cdots,q_{|I_C|}\},$ which satisfies
$\tilde{\sigma}(q_1,I)\geq \tilde{\sigma}(q_2,I)\geq\cdots \geq \tilde{\sigma}(q_{|I_C|},I),$ and
 $\gamma:=\lim_{n\rightarrow\infty}k/m$, $\alpha:=\lim_{n\rightarrow\infty}r/k$.
 \end{theorem}
\begin{proof}
See Appendix B.
\end{proof}

In Theorem \ref{thm:snr-back},  $\tilde{\sigma}_k(q_{k-r},I)$ increases when an initial support estimation has  more correct support because of the equation in which $\tilde{\sigma}_k(q,I)$ is given by the maximum value out of ${|I_C| \choose k-r}$ possibilities.  
Moreover, if we increase the ratio $m/(k+r)$, then the right-hand side of \eqref{snr-back} decreases  so that we can expect a greater possibility of  accurate support filtering with an increased redundant number of samples than the critical sampling rate.

To confirm a support filtering useful for performance improvement, we  examine the cases where the sufficient condition for an initial support estimation is less favorable than that of a support filtering.  Characterization of such  cases should be done with respect to a particular initial support estimation algorithm. 
For example,  in the case of a  subspace S-OMP for an initial support estimation,    a sufficient condition  for the success of subspace S-OMP for $\alpha>0$ is given by \cite{Kim2010CMUSIC}:
\begin{equation}\label{pert-omp}
\|P_{R(\tilde{U}_1)}-P_{R(U_1)}\|\ <
\frac{ \alpha- \alpha \sqrt{\gamma}(2-F(\alpha))}{2} \  ,
\end{equation}
where $F(\alpha)$ is an increasing function such that $F(1)=1$ and $\lim_{\alpha\rightarrow0^{+}}F(\alpha)=0$, which is defined as
$F(\alpha)=\frac{1}{\alpha}\int_0^{4t_1(\alpha)^2}xd\lambda_1(x),$
$d\lambda_1(x)=(\sqrt{(4-x)x})/(2\pi x)dx$
is the probability measure with support $[0,4]$, $0\leq t_1(\alpha)\leq 1$ satisfies
$\int_0^{2t_1(\alpha)}ds_1(x)=\alpha$,
and
$ds_1(x)=(1/\pi)\sqrt{4-x^2}dx$ is a probability measure with support $[0,2]$.
Now, the SNR condition for support filtering in Eq.~\eqref{snr-back} can be translated into a threshold of the allowable augmented subspace perturbation of $1-\gamma(1+\alpha)$, when $m\geq r+k$. Hence, the  gap between the two bounds  is given by   $$f(\gamma,\alpha) =  \frac{\alpha -\alpha \sqrt{\gamma}(2-F(\alpha))}{2} - (1-\gamma(1+\alpha) ).$$ 
Fig.~\ref{fig:alpha_gamma} characterizes the function $f(\gamma, \alpha)$. In region A,  $f(\gamma,\alpha)<0$ and $m>r+k$,  hence, the support filtering has more noise robustness and  can correct  errors from  subspace S-OMP.  As we can see from 
Fig.~\ref{fig:alpha_gamma}, support filtering is effective in most of the practical sampling rate, and especially when we have redundant samples or  ${\rm rank}(Y)$ is relatively small.
The larger size of region A that favors support filtering again confirms that support filtering is a quite useful technique to improve noise robustness.

\begin{figure}[htbp]
\centering\includegraphics[width=8cm]{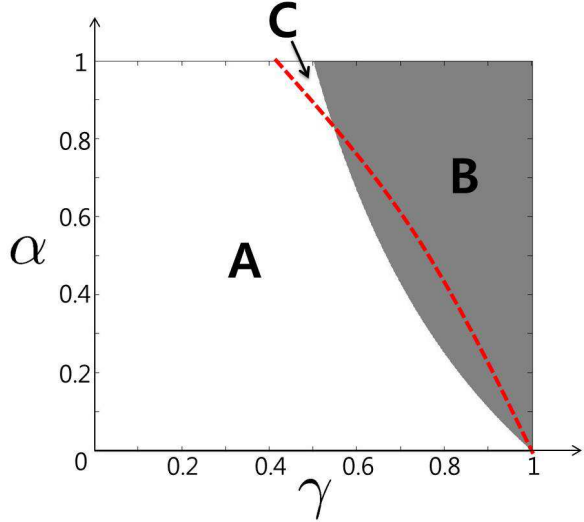}
\caption{Feasibility region where support filtering becomes more effective than subspace S-OMP.}
\label{fig:alpha_gamma}
\end{figure}

\section{NUMERICAL RESULTS}
\label{sec:results}

In this section, we perform extensive numerical experiments to validate the proposed algorithm  under various experimental conditions, and compare it with respect to  existing joint sparse recovery algorithms.

\subsection{Dependency on  Snapshot Number}

First, we demonstrate that a sequential CS-MUSIC is less sensitive to the number of snapshots. The simulation parameters were as follows: $m\in\{1,2,\ldots,30\}$, $n=128, k=8, r=4$, and $N\in\{6,16,256\}$, respectively. The elements of a sensing matrix $A$ were generated from a Gaussian distribution having zero mean and variance of $1/m$, and then each column of $A$ was normalized to have an unit norm.  An unknown signal $X$ with ${\rm rank}(X)=r\leq k$ was generated using the same procedure as in \cite{Lee2010SAMUSIC}. Specifically, we  randomly generated a support $I$, and then the corresponding nonzero signal components were obtained by
\begin{equation}\label{eq:source}
X^I=\Psi\Lambda\Phi \ ,
\end{equation} where $\Psi \in \Re^{k\times r}$ and $\Lambda$ were set to random orthonormal columns and the identity matrix, respectively, and $\Phi \in \Re^{r\times N}$ were made using Gaussian random distribution with zero mean and variance of $1/N$. After generating noiseless data,  we added zero mean white Gaussian noise to have $SNR=30dB$ measurements.
We declared success if an estimated support was the same as a true ${\rm supp}X$, and  success rates were averaged over $1000$ experiments. 

 Fig.~\ref{fig:Ndep} shows  success rates of a sequential CS-MUSIC compared to that of   CS-MUSIC or SA-MUSIC. Since SA-MUSIC in \cite{Lee2010SAMUSIC} is equivalent to CS-MUSIC for a normalized $A$ matrix,  the original code of SA-MUSIC was used for fair comparison.  As shown in Fig~\ref{fig:Ndep},  sequential CS-MUSIC exhibits nearly similar recovery performance for various snapshot numbers, whereas the original form of CS-MUSIC/SA-MUSIC requires a large  number of snapshots to achieve maximum performance. 

\begin{figure}[htbp]
\centering{\includegraphics[width=8cm]{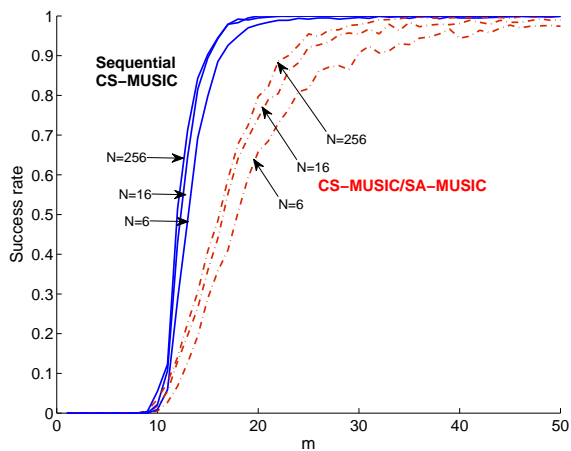}}
\caption{Snapshot dependent performance behaviour of the sequential CS-MUSIC and the original CS-MUSIC/SA-MUSIC.  The simulation parameters are $n=128, k=8, r=4$ and $SNR=30dB$. }
\label{fig:Ndep}
\end{figure}

In order to identify the contribution of the forward and backward greedy steps in the performance improvement, we   perform additional experiments  using the same simulation setup. Fig.~\ref{fig:each_step} illustrates the performances of sequential CS-MUSIC, a variation of sequential CS-MUSIC without backward support  filtering, and the  original CS-MUSIC/SA-MUSIC algorithm for $N=6,16$, respectively.  Here,  an initial $k-r$ support for CS-MUSIC/SA-MUSIC and the sequential CS-MUSIC were estimated using an identical   subspace S-OMP algorithm in \cite{Kim2010CMUSIC,Lee2010SAMUSIC} so that performance differences came only from the sequential subspace estimation step. For a bigger $N$ where the signal subspace error is  small,  performance improvement due to the sequential subspace estimation was not remarkable.  However, the advantages of sequential subspace estimation is especially noticeable for a small number of snapshots where a subspace estimation is prone to error.  On the other hand,  the backward support filtering is beneficial for all ranges of snapshots since it corrects the contribution of a partial support estimation error in a subspace S-OMP  step.

\begin{figure}[htbp]
\centering\includegraphics[width=8cm]{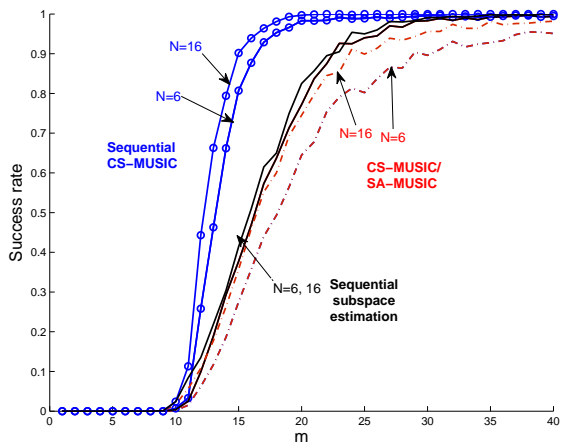}
\caption{Sequential CS-MUSIC, Sequential CS-MUSIC without greedy subspace estimation, and  CS-MUSIC/SA-MUSIC performance for various snapshot numbers.   Simulation parameters are $n=128, k=8, r=4$ and $SNR=30dB$. }
\label{fig:each_step}
\end{figure}

\subsection{Performance Comparison with  State-of-Art Joint Sparse Recovery Algorithms}

To compare the proposed algorithm with various state-of-art joint sparse recovery methods,  the recovery rates of various state-of-art joint sparse recovery algorithms such as   CS-MUSIC/SA-MUSIC, $l_1/l_2$ mixed norm approaches \cite{obozinski2011support,malioutov2005ssr,van2008probing},  and M-SBL \cite{wipf2007ebs}, are plotted in Fig.~\ref{fig:variousMMV} along with  those of a sequential CS-MUSIC.  Among the various implementation of mixed norm approaches, we used high performance  SGPL1 software \cite{van2008probing}, which can be downloaded from  ${\rm http://www.cs.ubc.ca/labs/scl/spgl1/}$. For M-SBL implementation, we used the original implementation by David Wipf. 
Since the mixed norm approach and M-SBL do not provide a exact $k$-sparse solution, we used the support for the largest $k$ coefficients as a  support estimate in calculating the perfect recovery ratio.
Figs.~\ref{fig:variousMMV}(a) and (b) show the recovery rates for $N=8$ and 256, respectively. Sequential  CS-MUSIC outperforms S-OMP and the original CS-MUSIC/SA-MUSIC consistently, and its performance nearly achieves those of M-SBL and the mixed norm approaches.  Note that the performance of M-SBL and the mixed norm approaches were identical.
Indeed, the additional sampling cost  for a  sequential CS-MUSIC compared to the M-SBL or the mixed norm approaches is  very small. Considering that any subspace method needs additional redundancy (i.e. $m\geq k+1$) to avoid ambiguity in the signal subspace estimation,  we believe that sequential CS-MUSIC nearly achieves the optimum performance.
Furthermore, this high performance can be achieved at negligible computational complexity.  Note that the complexity of the sequential CS-MUSIC is only a fraction of those of M-SBL and the mixed norm approaches,  as shown in  Figs.~\ref{fig:variousMMV}(c)(d) for $N=8$ and 256, respectively.   

\begin{figure}[htbp]
\centering{\includegraphics[width=8cm]{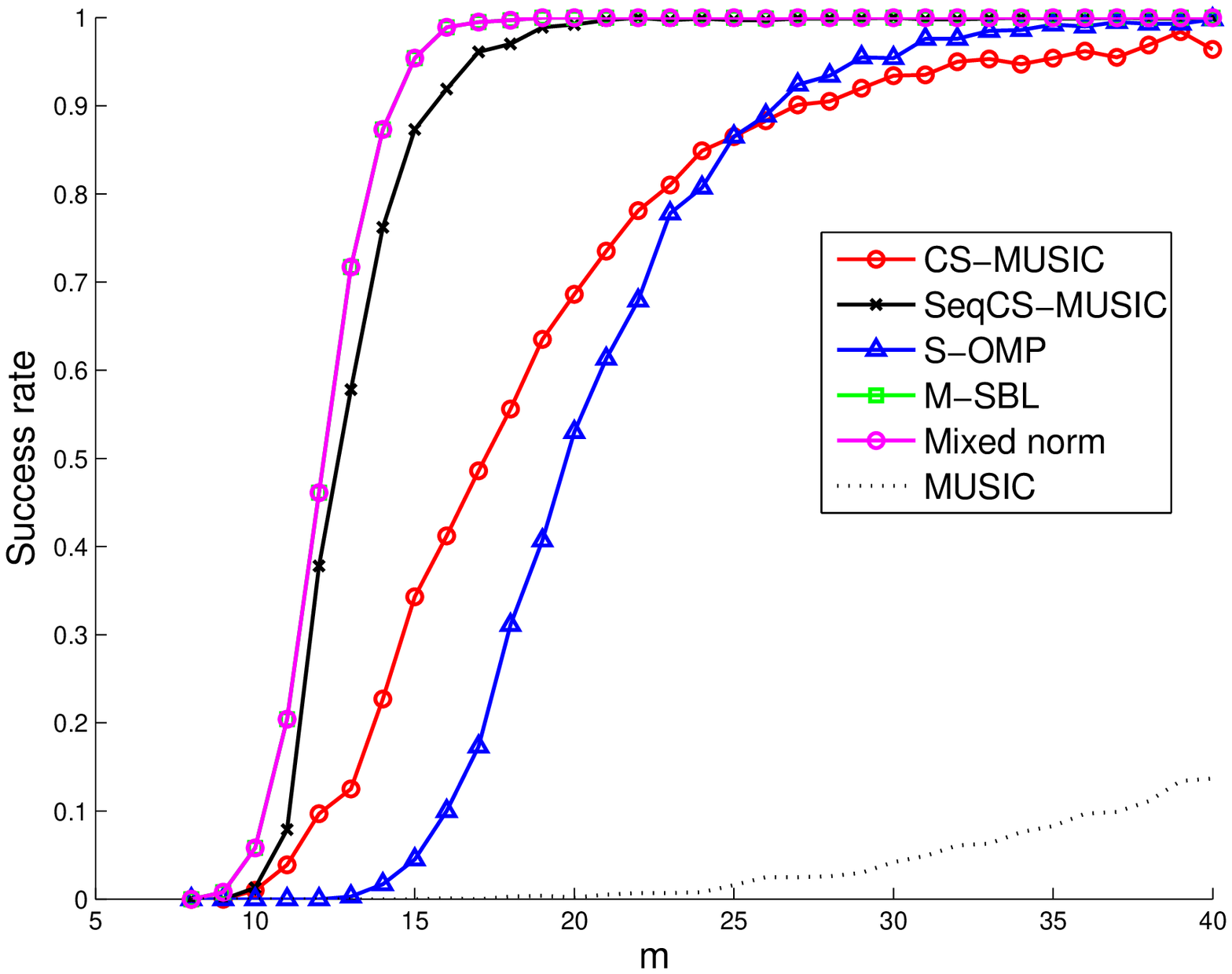}\includegraphics[width=8cm]{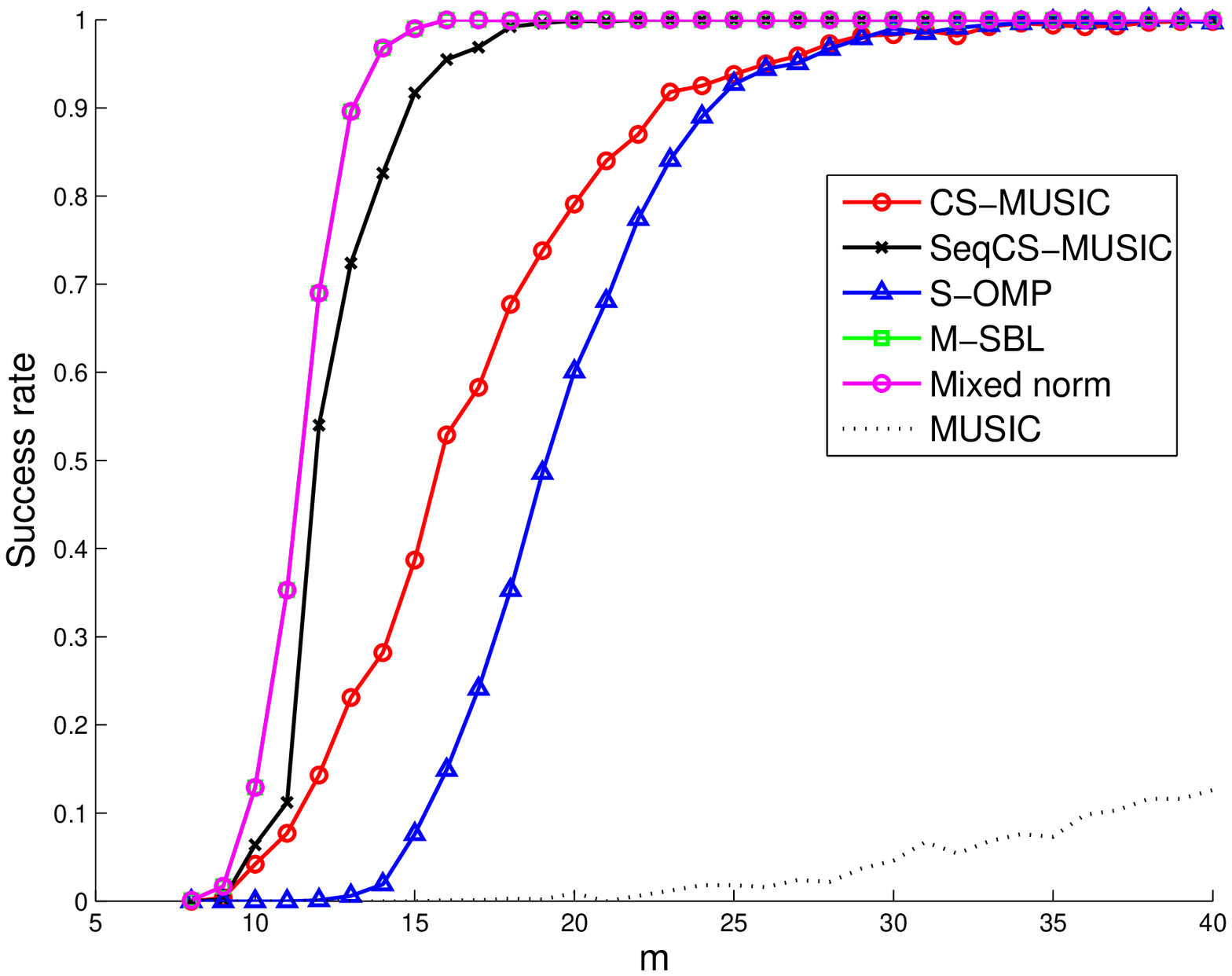}}
\centering{\includegraphics[width=8cm]{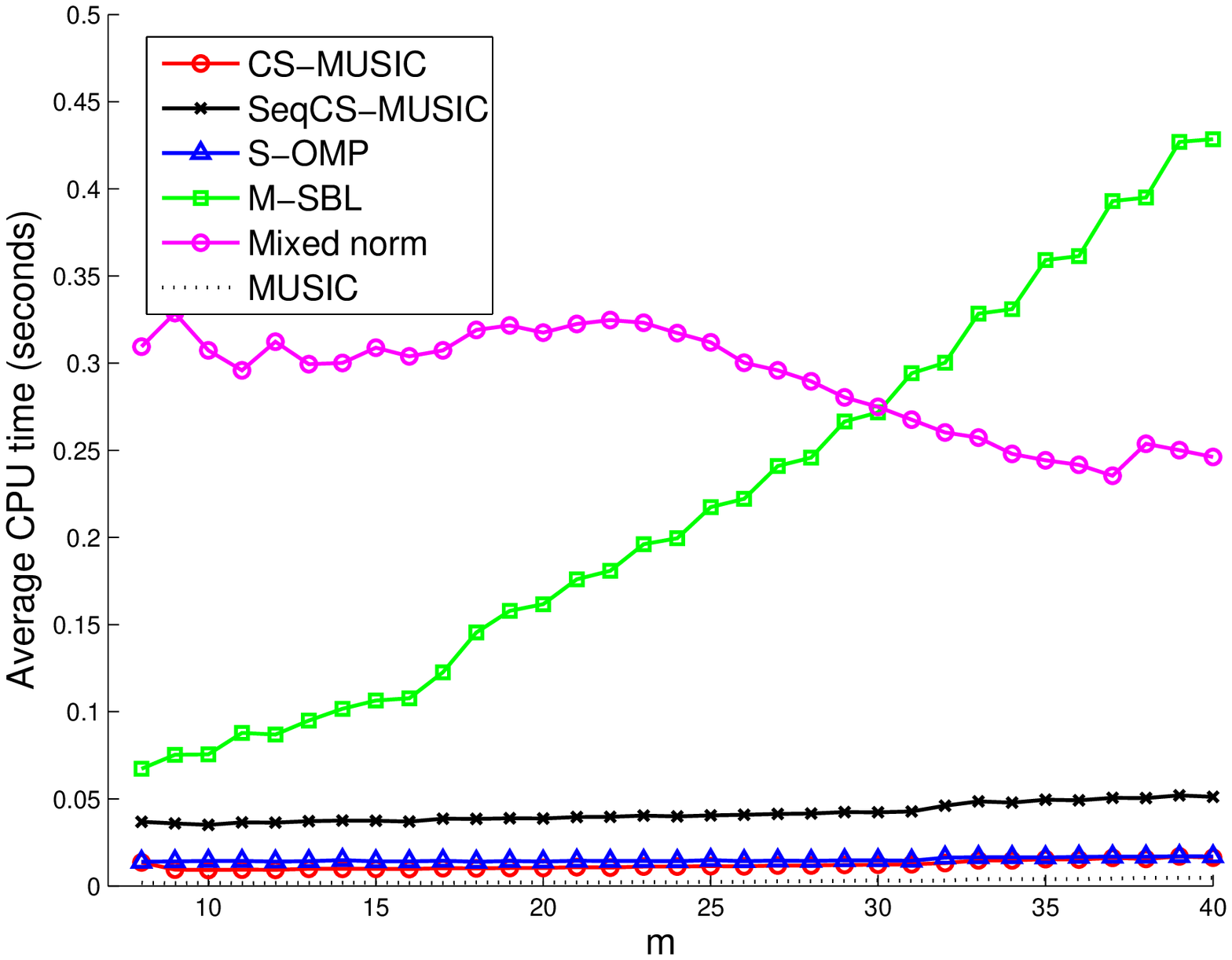}\includegraphics[width=8cm]{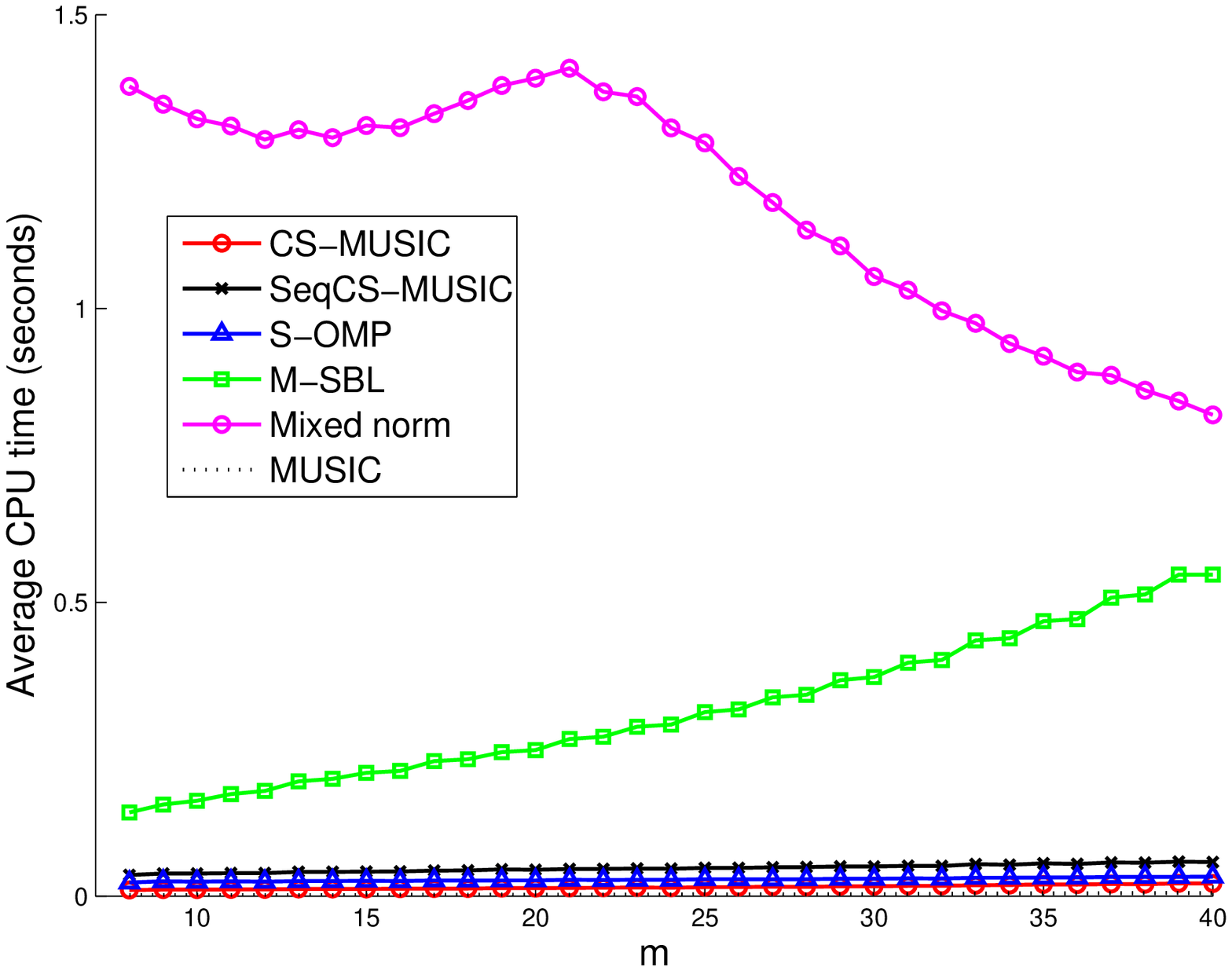}}
\caption{Performance of various joint sparse recovery algorithm at (a) $N=8$ and (b) $N=256$,   when $n=128, k=8, r=4$,  and $SNR=30dB$. (c)(d)  Average CPU time for $N=8$ and $N=256$, respectively.}
\label{fig:variousMMV}
\end{figure}

To show the dependency of recovery performance on the condition number of $X$, we conducted  simulations  for two different types of $X$. More specifically, the $j$-th diagonal term of $\Lambda$ in Eq.~\eqref{eq:source} is given by $\sigma_j=\tau^{j-1}$ for $j=1,\cdots,r$. The results in Fig.~\ref{fig:RIPcond}(a) provide evidence that sequential CS-MUSIC is not greatly affected by  the condition number of $X$, and  appears less sensitive than M-SBL.
Next, we performed  simulation studies  for  the different types of  RIP conditions using various MMV algorithms. More specifically, we assumed that each component of a sensing matrix follows $ \mathcal{N}(a, 1/m)$ and then  normalized each column of $A$ to have a unit norm.  The mean values are set to $a=0$ and $1$, where a larger $a$ represents a worse RIP condition. In this simulation, $N=64$ and the other parameters are the same as before.  Fig.~\ref{fig:RIPcond}(b) shows that sequential CS-MUSIC is more robust that the original CS-MUSIC/SA-MUSIC for unfavorable RIP conditions. 
However, compared to M-SBL, the sequential CS-MUSIC appears less robust to unfavorable RIP conditions, which is commonly observed in most of the greedy approaches.

\begin{figure}[htbp]
\centering{\includegraphics[width=8cm]{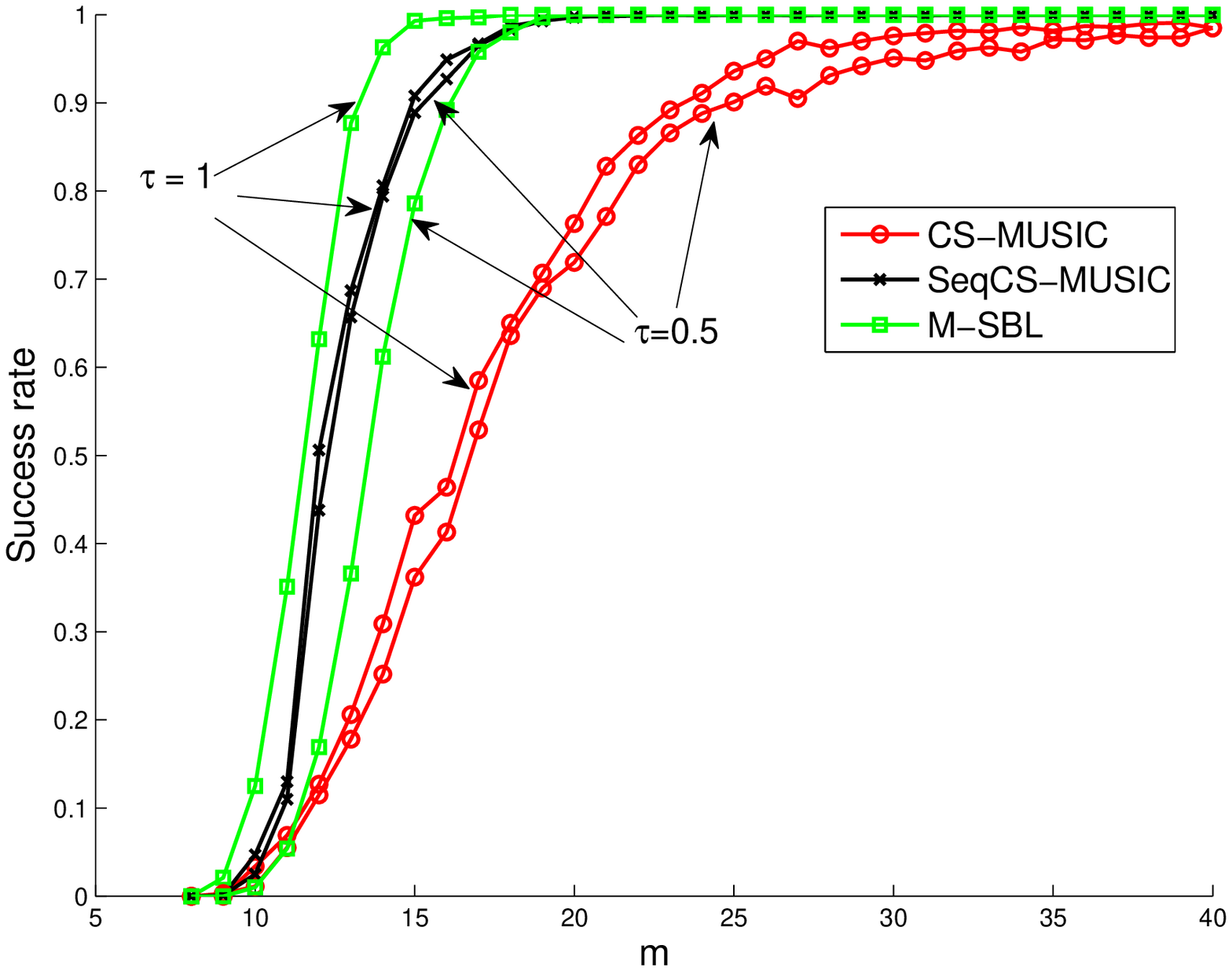}\includegraphics[width=8cm]{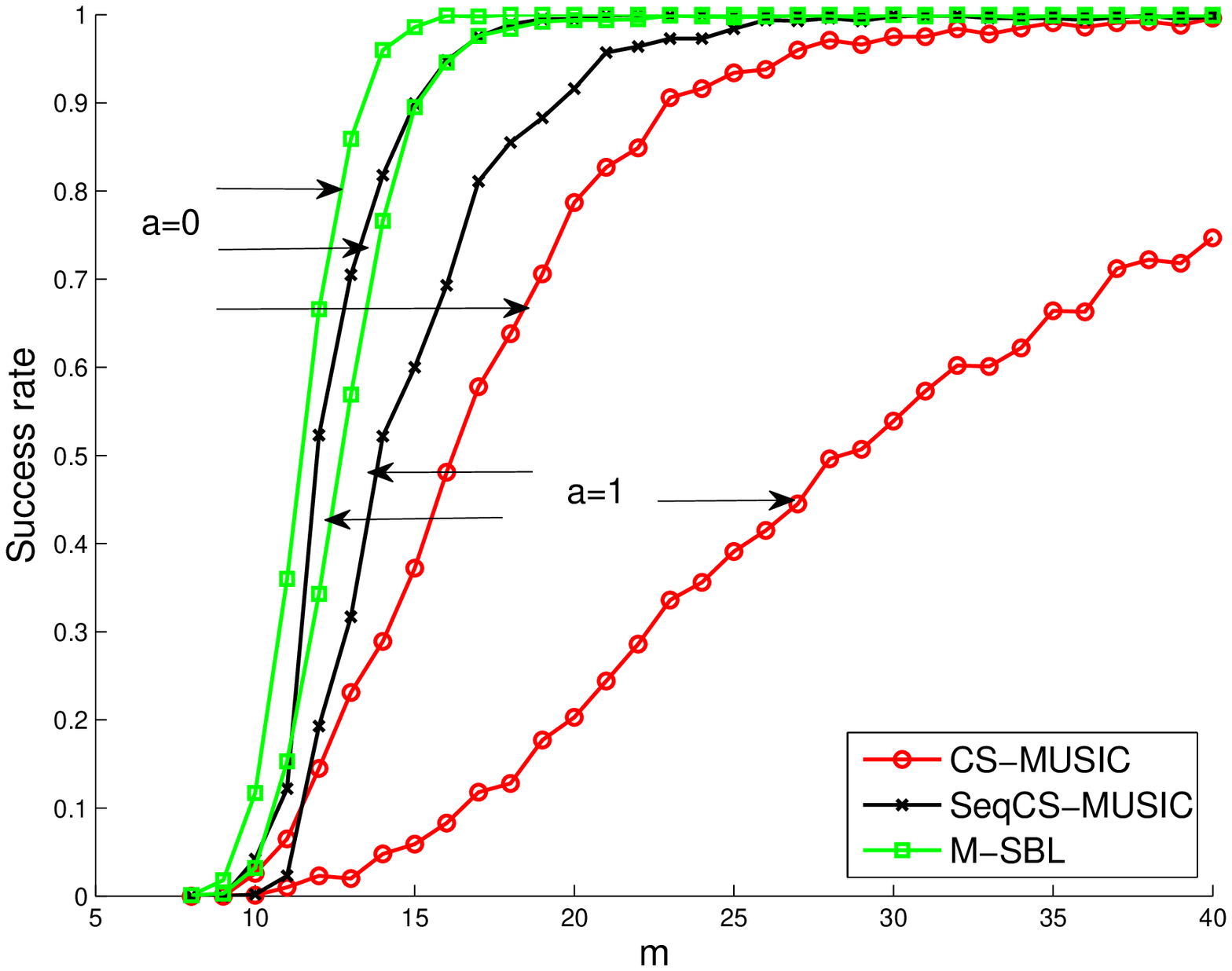}}
\caption{Recovery rates of M-SBL, sequential CS-MUSIC and CS-MUSIC/SA-MUSIC for (a) different condition numbers and (b) various RIP condition of $A$ (larger mean value provides worse RIP condition). The simulation parameters were $n=128, k=8, r=4$, $N=64$ and SNR=30dB.}
\label{fig:RIPcond}
\end{figure}

%

%
%

\subsection{Fourier Sensing Matrix}

Finally, we conducted similar numerical experiments using the Fourier sensing matrix. In this case,  the source model in Eq.~\eqref{eq:source} is  set to be complex valued.
Fig.~\ref{fig:variousMMVFourier} illustrates the recovery performance of various MMV algorithms for the Fourier sensing matrix when $n=128$, $k=8$, and $r=4$ at the SNR of 30dB for $N=5$.  We  again observed  similar performance improvement as in the Gaussian matrix.  
 Note again that the performance of M-SBL and the mixed norm approaches were identical.
However, compared to the Gaussian cases, a Fourier measurement  has  redundancies in imaginary information, which improves the overall recovery performance.

\begin{figure}[htbp]
\centering{\includegraphics[width=8cm]{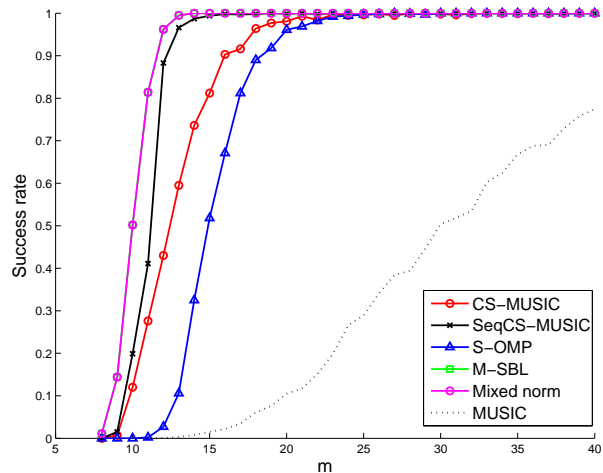}}
\caption{Performance of various joint sparse recovery algorithm from Fourier sensing matrix at    when $n=128, k=8, r=4$,  $N=5$ and $SNR=30dB$. }
\label{fig:variousMMVFourier}
\end{figure}

\section{DISCUSSION}

\subsection{Relation to Prior Constrained MUSIC and Sequential MUSIC}

In a prior constrained MUSIC algorithm \cite{linebarger1995incorporating}, the  prior knowledge of a direction-of-arrival (DOA) is incorporated into the MUSIC criterion to improve estimation performance  by filtering out the known source directions via orthogonal projections.  If the prior knowledge of a partial support is exact,  the prior constrained MUSIC is closely related to the generalized MUSIC step in CS-MUSIC/SA-MUSIC with an exact partial support estimate.
However,  as shown in \cite{wirfalt2011prior} as well as in this paper,  such an algorithm is affected by the resulting perturbation of the augmented subspace  if the number of snapshots  or SNR is not sufficiently high or the partial support knowledge is errorneous.   There have been  several approaches to improve the noise robustness of a prior constrained MUSIC  (see  \cite{wirfalt2011prior} and references therein); however,    to the best of our knowledge,   we are not aware of any existing method that improves the robustness of an augmented signal subspace estimate using sequential subspace estimation and support filtering.

Sequential MUSIC and its variations \cite{stoica1995improved,oh1993sequential,mosher1998recursive,mosher1999source} may appear  closely related to the proposed method. Indeed,  Davies {\em et al.}  \cite{davies2012rank} showed that  the 
Recursively Applied and Projected (RAP)-MUSIC \cite{mosher1999source}
is equivalent to a subspace S-OMP  (SS-OMP) step - a partial support recovery estimation part of CS-MUSIC. 
However, as shown in \cite{Kim2010CMUSIC}, the SNR requirement of a subspace S-OMP is much  tighter than that of a  generalized MUSIC step.  Therefore, switching from sequential MUSIC to the generalized MUSIC step would be beneficial.
Moreover,   our analysis in this paper showed that the  sequential subspace estimation further relaxes the  SNR condition sequentially. This implies that   even if 
the SNR condition  of the generalized MUSIC  is not satisfied initially, during the sequential subspace estimation the SNR condition can be met and the overall recovery performance can be improved.   To the best of our knowledge,  this type of techniques have not been reported for any existing sequential MUSIC algorithms  \cite{stoica1995improved,oh1993sequential,mosher1998recursive,mosher1999source}.

\subsection{Sparsity Estimation}

So far,  our derivation assumes the prior knowledge of support size.  In our previous work \cite{Kim2010CMUSIC},  we derived  a sparsity estimation algorithm. Note that  this algorithm can be incorporated at each greedy step to make the algorithm work, even without knowing the sparsity level {\em a priori}. In addition, there are various heuristics that could be
used to estimate the sparsity in  MUSIC type parametric methods.    However,  they have a nonzero probability of being wrong if the measurement is noisy. Typically, these sparsity estimation algorithms tend to overestimate,  so various model order selection criteria have been often incorporated to avoid this overestimation \cite{stoica2004model}.  

Note that  the RIP condition $\delta^{L}_{2k-r+l}<1$ implies that  the maximal sparsity level that our algorithm can recover does not exceed $m-1$, and the corresponding submatrix $A_{\hat I}$ for the support estimate $\hat I$ is always full column ranked.   
This implies that, as long as the condition number of $A_{\hat I}$ is not bad and the noise levels are sufficiently small, 
 we can implement thresholding techniques in a reconstruction domain to find a sparse signal,  similar to other non-parametric sparse recovery approaches like iterative thresholding, M-SBL, etc. 
When such thresholding scheme may not be sufficiently accurate,   the current technique has limitations and we need a new way to estimate the sparsity level. Though the sparsity estimation is very important topic, this is  beyond scope of the current work, and will be reported elsewhere.
%

\subsection{Limitation of Noisy Analysis}
 
Even though our noisy analysis provides useful insight on the origin of the noise robustness of a sequential CS-MUSIC,  current analysis has two limitations. First, the analysis is based on the Gaussian sensing matrix using an asymptotic argument.  Hence,  the analysis should be modified for a general sensing matrix such as Fourier. An RIP based analysis in SA-MUSIC \cite{Lee2010SAMUSIC} would work toward this goal.  Second, the noisy performance analysis is based on comparing sufficient conditions. Since a sufficient condition is often more restricted than necessary,  the analysis in this paper should be  understood as a more conservative comparison.

\section{CONCLUSION}
\label{sec:conclusion}

In this paper, we derived  two greedy strategies to improve the noise robustness of 
recent hybrid joint sparse recovery algorithms such as CS-MUSIC and SA-MUSIC. Although these hybrid algorithms significantly outperform any other conventional greedy MMV algorithms, the performance improvement is reduced for a limited number of snapshots.
We showed that the performance degradation is due to a perturbation in an augmented signal subspace estimation originating from an inaccurate subspace or
partial support estimation.  Furthermore,  we demonstrated that  even with limited number of snapshots, there are two different ways to improve the noise robustness of augmented signal subspace estimation: one by sequential subspace estimation and the other by  filtering out incorrect support.   
We further explained that the two greedy steps are byproducts of 
a novel generalized subspace criterion.  Theoretical analysis in noisy situations revealed the origins of the noise robustness of the proposed algorithm and 
led to the identification of  sampling conditions where each greedy step  becomes beneficial.
Extensive numerical simulation demonstrated that the new algorithm consistently outperforms  the existing greedy algorithms  and nearly achieves optimal performance with minimal computational complexity.

\section*{Appendix~A}
To obtain the perturbation bound Eq.~\eqref{eq:upper}  in sequential subspace estimation, we use the following theorem. 
\begin{theorem}\label{thm:perturb}\cite{Wedin72pbsvd}
Assume that $G\in\mathbb{R}^{m\times q}$ has the singular value decomposition
$$G=U\Sigma V^{*}=U_1\Sigma_1V_1^{*}+U_0\Sigma_0V_0^{*}=G_1+G_0$$
where $G_1:=U_1\Sigma_1V_1^{*}$ and $G_0:=U_0\Sigma_0V_0^{*}.$
Also, for a perturbed matrix $\tilde{G}\in\mathbb{R}^{m\times q}$ of $A$, assume that $\tilde{G}$ has the singular value decomposition
$$\tilde{G}=\tilde{U}\tilde{\Sigma}\tilde{V}^{*}=\tilde{U}_1\tilde{\Sigma}_1\tilde{V}_1^{*}+\tilde{U}_0\tilde{\Sigma}_0\tilde{V}_0^{*}=\tilde{G}_1+\tilde{G}_0$$
where $\tilde{G}_1:=\tilde{U}_1\tilde{\Sigma}_1\tilde{V}_1$ and $\tilde{G}_0:=\tilde{U}_0\tilde{\Sigma}_0\tilde{V}_0^{*}$, and $U_1$ and $\tilde{U}_1$ (or $U_0$ and $\tilde{U}_0$) are the matrices of same size. If there exist $\alpha\geq 0$ and $\delta>0$ such that
\begin{equation*}\label{sin-cond}
\sigma_{\min}(\tilde{G}_1)\geq \alpha+\delta~{\rm and}~\sigma_{\max}(G_0)\leq \alpha,
\end{equation*}
then for every unitary invariant norm,
\begin{equation*}\label{sin-bound}
\sin{\theta(R(\tilde{G}_1),R(G_1))}=\|P_{R(\tilde{G}_1)}-P_{R(G_1)}\|\leq \frac{\epsilon}{\delta},
\end{equation*}
where 
\begin{align*}
\epsilon:=\max(\|R_1\|,\|R_2\|)&,&
R_1:=-(\tilde{G}-G)V_1,R_2:=-(\tilde{G}-G)^{*}U_1.
\end{align*}
\end{theorem}

{\em (Proof of Eq.~\eqref{eq:upper})}
For a noiseless measurement $[A_{I_{k-r+l}}~S]$, we have
$$\sigma_1\geq\cdots\geq\sigma_k>\sigma_{k+1}=\cdots=\sigma_{k+l}=0$$
so that  we have $\sigma_{\min}(\tilde{U}_1\tilde{\Sigma}_1\tilde{V}_1^{*})\geq \sigma_k([A_{I_{k-r+l}}~S])-\|S-\tilde{S}\|$ and $\sigma_{\max}(U_0\Sigma_0V_0^{*})=0$ so that if
we have $\Delta:=\|P_{R(S)}-P_{R(\tilde{S})}\|<\sigma_k([A_{I_{k-r+l}}~S])$, we can apply Theorem \ref{thm:perturb}. Then, we have
$$\frac{\epsilon}{\delta}=
\frac{\max{(\|(S-\tilde{S})V_1\|,\|(S-\tilde{S})^{*}U_1\|)}}{\sigma_k([A_{I_{k-r+l}}~B])-\|S-\tilde{S}\|}\leq
\frac{\Delta}{\sigma_k([A_{I_{k-r+l}}~S])-\Delta}$$
if $\Delta<\sigma_k([A_{I_{k-r+l}}~S]).$  This concludes the proof.

\section*{Appendix~B}
\label{ap:B}
By the assumption $|I\cap {\rm supp}X|\geq k-r+1$ and the generalized MUSIC criterion, we have $\|P_{R([A_{I\setminus\{q\}}~S])}^{\perp}\ab_q\|^2=0$ for any $j\in {\rm supp}X$. Then, for any $j\notin{\rm supp}X$ and $q\in {\rm supp}X$, we have
\begin{eqnarray}\label{apb-eq1}
&&\|P_{R([A_{I\setminus \{j\}}~\tilde{S}])}^{\perp}\ab_j\|^2-\|P_{R([A_{I\setminus\{q\}}~\tilde{S}])}^{\perp}\ab_q\|^2\notag\\
&\geq&\|P_{R([A_{I\setminus \{j\}}~\tilde{S}])}^{\perp}\ab_j\|^2-\|P_{R([A_{I_{k-r,q}}~\tilde{S}])}^{\perp}\ab_q\|^2\notag\\
&=&\|P_{R([A_{I\setminus \{j\}}~\tilde{S}])}^{\perp}\ab_j\|^2-\|P_{R([A_{I_{k-r,q}}~\tilde{S}])}^{\perp}\ab_q\|^2\notag \\
&&+\|P_{R([A_{I_{k-r,q}}~S])}^{\perp}\ab_q\|^2\notag\\
&\geq&\|P_{R([A_{I\setminus \{j\}}~\tilde{S}])}^{\perp}\ab_j\|^2-\|\ab_q\|^2
\|P_{R([A_{I_{k-r,q}}~\tilde{S}])}-P_{R([A_{I_{k-r,q}}~S])}\|,\notag\\
\end{eqnarray}
since $P^2=P$ for any orthogonal projection operator $P$, where $I_{k-r,q}\subset (I\setminus \{q\})\cap {\rm supp}X.$
For $j\notin {\rm supp}X$, $\ab_j$ is statistically independent from $R([A_{I\setminus\{j\}}~\tilde{S}])$ so that $m\|P_{R([A_{I\setminus\{j\}}~\tilde{S}])}^{\perp}\ab_j\|^2$ is chi-squared random variable with at least $m-k-r+1$ degrees of freedom so that
\begin{eqnarray}\label{apb-eq2}
\liminf\limits_{n\rightarrow\infty}\min\limits_{j\notin {\rm supp}X}\ab_j^{*}
P_{R([A_{I\setminus\{j\}}~\tilde{S}])}^{\perp}\ab_j
\geq 1-\gamma(1+\alpha),
\end{eqnarray}
where $\gamma:=\lim_{n\rightarrow\infty}k/m$ and $\alpha:=\lim_{n\rightarrow\infty}r/k$.

On the other hand, for any $q\in {\rm supp}X$, $m\|\ab_q\|^2$ is a chi-squared random variable of $m$ degrees of freedom, so that by Lemma 3 in \cite{Fletcher2009nscond} we have
$\lim\limits_{n\rightarrow\infty}\max\limits_{q\in {\rm supp}X}\|\ab_q\|^2=1.$
Since $[A_{I_{k-r,q}}~S]$ has a full column rank, using the bound in Eq.~\eqref{eq:upper} 
for any $I_{k-r,q}\subset ({\rm supp}X\cap I)\setminus \{q\}$, we have 
\begin{eqnarray*}
\|P_{R([A_{I_{k-r,q}}~S])}-P_{R([A_{I_{k-r,q}}~\tilde{S}])}\|
\leq
 \frac{\Delta}{\max\limits_{T\subset ({\rm supp}X\cap I)\setminus\{q\}}\sigma_k([A_{I_T}~S])-\Delta}.
\end{eqnarray*}
If we let $\tilde{\sigma}(q,I)=\max\{\sigma_k([A_{I_T}~S]):T\subset I_C\setminus\{q\},|T|=k-r\}$, where $I_C=\{q_1,\cdots,q_{|I_C|}\}$ and 
$$\tilde{\sigma}(q_1,I)\geq \tilde{\sigma}(q_2,I)\geq\cdots \geq \tilde{\sigma}(q_{|I_C|},I),$$
we have 
\begin{equation}\label{apb-eq3}
\|P_{R([A_{I_{k-r,q}}~S])}-P_{R([A_{I_{k-r,q}}~\tilde{S}])}\|\leq 
\frac{\Delta}{\tilde{\sigma}_k(q_{k-r},I)-\Delta}
\end{equation}
for $q=q_1,\cdots,q_{k-r}$. Hence, by \eqref{apb-eq1}, \eqref{apb-eq2} and \eqref{apb-eq3}, \eqref{back-success} holds if we have
$$1-\gamma(1+\alpha)-\frac{\Delta}{\tilde{\sigma}_k(q_{k-r},I)-\Delta}>0$$
in the large system limit. This completes the proof.

\section*{Acknowledgment}

This work was supported by the Korea Science and Engineering Foundation (KOSEF) grant funded by the Korea government (MEST) (No.2011-0000353).
The authors would like to thank  Kiryung Lee and Yoram Bresler for helpful discussions and providing an SA-MUSIC code.

\ifCLASSOPTIONcaptionsoff
  \newpage
\fi

\bibliographystyle{IEEEbib}
\bibliography{totalbiblio_bispl}

\begin{thebibliography}{10}

\bibitem{KrVi96}
H.~Krim and M.~Viberg,
\newblock ``Two decades of array signal processing research,''
\newblock {\em IEEE Signal Proc. Magazine}, pp. 67--94, July 1996.

\bibitem{PrWeScBo99}
K.~P. Pruessmann, M.~Weigher, M.~B. Scheidegger, and P.~Boesiger,
\newblock ``{SENSE}: Sensitivity encoding for fast {{MRI}},''
\newblock {\em Magn. Reson. Med}, vol. 42, no. 5, pp. 952--962, 1999.

\bibitem{Joshi2006spi}
A.~Joshi, W.~Bangerth, and E.~M. Sevick-Muraca,
\newblock ``Non-contact fluorescence optical tomography with scanning patterned
  illumination,''
\newblock {\em Optics Express}, vol. 19, no. 14, pp. 6516--6534, July 2006.

\bibitem{Lee2011CDOT}
O.~K. Lee, J.~M. Kim, Y.~Bresler, and J.~C. Ye,
\newblock ``Compressive diffuse optical tomography: non-iterative exact
  reconstruction using joint sparsity,''
\newblock {\em IEEE Trans. Med. Imag.}, vol. 30, no. 5, pp. 1129--1142, 2011.

\bibitem{BaWaDuSaBa05}
D.~Baron, M.B. Wakin, M.F. Duarte, S.~Sarvotham, and R.G. Baraniuk,
\newblock ``{Distributed compressed sensing},''
\newblock {\em preprint}, 2005.

\bibitem{Duarte2006IPSN}
M.F. Duarte, M.B. Wakin, D.~Baron, and R.G. Baraniuk,
\newblock ``Universal distributed sensing via random projections,''
\newblock in {\em Proceedings of the International Conference on Information
  Processing in Sensor Networks}, Nashville, TN, 2006, pp. 177--185.

\bibitem{Kim2010CMUSIC}
J.M. Kim, O.K. Lee, and J.C. Ye,
\newblock ``{Compressive MUSIC: revisiting the link between compressive sensing
  and array signal processing},''
\newblock {\em IEEE Trans. on Information Theory}, vol. 58, no. 1, pp.
  278--301, 2012.

\bibitem{chen2006trs}
J.~Chen and X.~Huo,
\newblock ``{Theoretical results on sparse representations of multiple
  measurement vectors},''
\newblock {\em IEEE Trans. on Signal Processing}, vol. 54, no. 12, pp.
  4634--4643, 2006.

\bibitem{cotter2005ssl}
S.F. Cotter, B.D. Rao, K.~Engan, and K.~Kreutz-Delgado,
\newblock ``{Sparse solutions to linear inverse problems with multiple
  measurement vectors},''
\newblock {\em IEEE Trans. on Signal Processing}, vol. 53, no. 7, pp. 2477,
  2005.

\bibitem{Mishali08rembo}
M.~Mishali and Y.~C. Eldar,
\newblock ``{Reduce and boost: Recovering arbitrary sets of jointly sparse
  vectors},''
\newblock {\em IEEE Trans. on Signal Processing}, vol. 56, pp. 4692--4702,
  2009.

\bibitem{Berg09jrmm}
E.~Berg and M.~P. Friedlander,
\newblock ``Theoretical and empirical results for recovery from multiple
  measurements,''
\newblock {\em IEEE Trans. on Information Theory}, vol. 56, no. 5, pp.
  2516--2527, 2010.

\bibitem{Lee2010SAMUSIC}
K.~Lee, Y.~Bresler, and M.~Junge,
\newblock ``Subspace methods for joint sparse recovery,''
\newblock {\em IEEE Trans. on Information Theory (in press)}.

\bibitem{obozinski2011support}
G.~Obozinski, M.J. Wainwright, and M.I. Jordan,
\newblock ``Support union recovery in high-dimensional multivariate
  regression,''
\newblock {\em The Annals of Statistics}, vol. 39, no. 1, pp. 1--47, 2011.

\bibitem{Sc86}
R.~Schmidt,
\newblock ``Multiple emitter location and signal parameter estimation,''
\newblock {\em IEEE Trans. Antennas and Propagation}, vol. 34, no. 3, pp.
  276--280, 1986.

\bibitem{linebarger1995incorporating}
D.A. Linebarger, R.D. DeGroat, E.M. Dowling, P.~Stoica, and G.L. Fudge,
\newblock ``{Incorporating a priori information into MUSIC-algorithms and
  analysis},''
\newblock {\em Signal Processing}, vol. 46, no. 1, pp. 85--104, 1995.

\bibitem{wirfalt2011prior}
P.~Wirfalt, M.~Jansson, G.~Bouleux, and P.~Stoica,
\newblock ``Prior knowledge-based direction of arrival estimation,''
\newblock in {\em IEEE International Conference on Acoustics, Speech and Signal
  Processing (ICASSP}. IEEE, 2011, pp. 2540--2543.

\bibitem{wipf2007ebs}
DP~Wipf and BD~Rao,
\newblock ``{An Empirical {B}ayesian Strategy for Solving the Simultaneous
  Sparse Approximation Problem},''
\newblock {\em IEEE Trans. on Signal Processing}, vol. 55, no. 7 Part 2, pp.
  3704--3716, 2007.

\bibitem{stoica2012LIKES}
P.~Stoica and P~Babu,
\newblock ``{SPICE and LIKES : Two hyper-parameter free methods for
  sparse-parameter estimation},''
\newblock {\em Signal Processing (in press)}, 2012.

\bibitem{stoica2011spice}
P.~Stoica, P.~Babu, and J.~Li,
\newblock ``{SPICE: A sparse covariance-based estimation method for array
  processing},''
\newblock {\em IEEE Trans. on Signal Processing}, vol. 59, no. 2, pp. 629--638,
  2011.

\bibitem{malioutov2005ssr}
D.~Malioutov, M.~Cetin, and AS~Willsky,
\newblock ``{A sparse signal reconstruction perspective for source localization
  with sensor arrays},''
\newblock {\em IEEE Trans. on Signal Processing}, vol. 53, no. 8, pp.
  3010--3022, 2005.

\bibitem{Reeves2008ISIT}
G.~Reeves and M.~Gastpar,
\newblock ``Sampling bounds for sparse support recovery in the presence of
  noise,''
\newblock in {\em Proceedings of the IEEE International Symposium of
  Information Theory (ISIT 2008)}, Toronto, Canada, 2008, pp. 2187--2191.

\bibitem{Fletcher2009nscond}
S.~Rangan A.K.~Fletcher and V.K. Goyal,
\newblock ``Necessary and sufficient conditions for sparsity pattern
  recovery,''
\newblock {\em IEEE Trans. on Inform. Theory}, vol. 55, no. 12, pp. 5758--5772,
  December 2009.

\bibitem{van2008probing}
E.~Van Den~Berg and M.P. Friedlander,
\newblock ``Probing the pareto frontier for basis pursuit solutions,''
\newblock .

\bibitem{stoica1995improved}
P.~Stoica, P.~Handel, and A.~Nehoral,
\newblock ``Improved sequential music,''
\newblock {\em IEEE Trans. on Aerospace and Electronic Systems}, vol. 31, no.
  4, pp. 1230--1239, 1995.

\bibitem{oh1993sequential}
S.K. Oh and C.K. Un,
\newblock ``A sequential estimation approach for performance improvement of
  eigenstructure-based methods in array processing,''
\newblock {\em IEEE Trans. on Signal Processing}, vol. 41, no. 1, pp. 457,
  1993.

\bibitem{mosher1998recursive}
J.C. Mosher and R.M. Leahy,
\newblock ``{Recursive MUSIC: a framework for EEG and MEG source
  localization},''
\newblock {\em IEEE Trans. on Biomedical Engineering}, vol. 45, no. 11, pp.
  1342--1354, 1998.

\bibitem{mosher1999source}
J.C. Mosher and R.M. Leahy,
\newblock ``Source localization using recursively applied and projected (rap)
  music,''
\newblock {\em IEEE Trans. on Signal Processing}, vol. 47, no. 2, pp. 332--340,
  1999.

\bibitem{davies2012rank}
M.E. Davies and Y.C. Eldar,
\newblock ``Rank awareness in joint sparse recovery,''
\newblock {\em IEEE Trans. on Information Theory}, vol. 58, no. 2, pp.
  1135--1146, 2012.

\bibitem{stoica2004model}
P.~Stoica and Y.~Selen,
\newblock ``Model-order selection: a review of information criterion rules,''
\newblock {\em IEEE Signal Processing Magazine}, vol. 21, no. 4, pp. 36--47,
  2004.

\bibitem{Wedin72pbsvd}
P.-{\AA}. Wedin,
\newblock ``Perturbation bounds in connection with singular value
  decomposition,''
\newblock {\em BIT}, pp. 99--111, 1972.

\end{thebibliography}

\end{document}